\documentclass[JHEP,12pt]{article}

\usepackage{booktabs}

\usepackage{jheppub}
\usepackage{amssymb,amsmath,amsthm}
\usepackage[shortlabels]{enumitem}
\usepackage{braket}
\usepackage{breqn,hyperref}
\usepackage{comment} 
\usepackage{subcaption}
\usepackage{graphicx, xcolor, varwidth}
\usepackage{color}
\newtheorem{thm}{Theorem}%[section]
\theoremstyle{definition}
\newtheorem{defn}[thm]{Definition}

\theoremstyle{remark}
\newtheorem{rem}[thm]{Remark}

\DeclareMathOperator{\setint }{int}

\DeclareMathOperator{\A}{Area}

\DeclareMathOperator{\EW}{EW}

\title{Holographic Entropy Cone Beyond AdS/CFT}
\author{Raphael Bousso and Sami Kaya}

\affiliation{Center for Theoretical Physics and Department of Physics,\\
University of California, Berkeley, California 94720, U.S.A. 
} 

\emailAdd{bousso@berkeley.edu}
\emailAdd{samikaya@berkeley.edu}

\abstract{We extend all known area inequalities obeyed by Ryu-Takayanagi surfaces of AdS boundary regions -- the holographic entropy cone -- to static generalized entanglement wedges of bulk regions in arbitrary spacetimes. The generalized holographic entropy cone is subject to a mutual independence condition on the bulk regions: each bulk input region must be outside the entanglement wedge of the union of all others. The condition captures when gravitating regions involve fundamentally distinct degrees of freedom despite the nonlocality inherent in the holographic principle.} 

\makeatletter
\gdef\@fpheader{\mbox{}}
\makeatother

\begin{document}
\maketitle
\section{Introduction}

\paragraph{Entanglement Wedges} In recent years, significant progress has been made in understanding the deep connection between quantum entanglement and spacetime geometry in the context of holography, particularly through the concept of entanglement wedges. In Anti-de Sitter/Conformal Field Theory (AdS/CFT)  duality~\cite{Maldacena:1997re}, the entanglement wedge $\EW(B)$ for a given boundary subregion $B$~\cite{Ryu:2006bv,Hubeny:2007xt,Faulkner:2013ana,Engelhardt:2014gca} corresponds to a bulk region whose degrees of freedom are encoded within $B$, encapsulating the holographic principle. At leading order in Newton's constant $G$ (or in the inverse of the central charge $N$ of the dual CFT), the entanglement wedge is the homology region between $B$ and the smallest-area extremal surface homologous to $B$~\cite{Ryu:2006bv,Hubeny:2007xt}.

In a large class of bulk states~\cite{Akers:2020pmf}, this area computes the von Neumann entropy of the corresponding boundary subregions:
\begin{equation}
    S_B = \frac{\A[\EW(B)]}{4G} + O(G^0)~.
\end{equation}
The von Neumann entropies of distinct boundary regions must satisfy any universal inequality that holds for the entropy of subsystems in quantum mechanics. An example is the strong subadditivity (SSA) condition, which holds for all tripartite quantum systems: 
\begin{equation}
    S_{AB}+S_{BC}\geq S_{ABC}+S_B~.
\end{equation}
There is no reason a priori that the areas of geometric surfaces obey the same relation. Thus, the proof~\cite{Headrick:2007km,Wall:2012uf} that the areas of entanglement wedges do obey it,
\begin{equation}
    \A[\EW(AB)]+\A[\EW(BC)]\geq \A[\EW(ABC)]+\A[\EW(B)]~,
\end{equation}
marked an important consistency check for entanglement wedge duality. This agreement can be extended to higher order in $G$ by appealing to SSA of the quantum state of matter fields in bulk regions. 

\paragraph{Holographic Entropy Cone} The areas of entanglement wedges obey are additional inequalities, collectively known as the holographic entropy cone (HEC)~\cite{Bao:2015bfa}. These do not extend to higher order in $G$, and so they need not hold for general quantum states or even for general states of CFT subregions. They constrain the entanglement entropy of the CFT only at leading order in $G$ or $1/N$. Thus, they isolate the entanglement structure of the boundary states associated directly to the classical bulk geometry, as opposed to the bulk matter fields. 

The best known HEC inequality is the monogamy of mutual information (MMI)~\cite{Hayden:2011ag},
\begin{equation}\label{MMI}
    S_{AB}+S_{BC}+S_{AC}\geq S_{A}+S_{B}+S_{C}+S_{ABC}~.
\end{equation}
The next simplest is the five-party cyclic inequality,
\begin{equation} \label{five-party-cyclic}
    S_{ABC}+S_{BCD}+S_{CDE}+S_{DEA}+S_{EAB}\geq S_{AB}+S_{BC}+S_{CD}+S_{DE}+S_{EA}+S_{ABCDE}~.
\end{equation}
Large classes of HEC inequalities have been proven ``by contraction,'' a graph-theory technique first introduced in Ref.~\cite{Bao:2015bfa}\footnote{The unique contraction map for MMI \ref{MMI} and a non-unique contraction map for five-party cyclic inequality \ref{five-party-cyclic} was also explicitly given in \cite{Bao:2015bfa}. }. (Below will follow the more accessible presentation in Ref.~\cite{Akers:2021lms}.)
% \raphael{I feel like we should say more about how large these classes are and whether all HEC inequalities can be proven this way and whether all are known. We should also refer the reader to known contraction maps, here or below. Can you try doing this?} 
Indeed, the problem of finding all inequalities is equivalent to finding all contraction maps \cite{Bao:2024azn}. Our theorem will assume only the existence of a contraction map, so it will establish that the HEC and the generalized HEC are identical.

Despite much progress in recent years in enumerating and characterizing these inequalities, the full set of HEC inequalities remains unknown. This is mainly because the number of inequalities grows rapidly with the number of regions. To date, the HEC is fully known only for up to five parties \cite{HernandezCuenca:2019wgh}.  Nevertheless, thousands of six-party HEC inequalities \cite{Hernandez-Cuenca:2023iqh} and two infinite families involving an arbitrarily high number of parties \cite{Czech:2023xed} are also known. A first step in classifying all HEC inequalities by studying the relevant contraction maps was taken in Ref.~\cite{Bao:2024azn}.  

With the exceptions of SSA and MMI, the HEC inequalities have so far been proven only in static settings, where the geometry of the bulk is time-independent.
For some attempts at generalizing the holographic entropy cone to time-dependent settings in 2+1 dimensions see Refs.~\cite{Czech_2019, Grado-White:2024gtx}. Many distinct HEC inequalities were also tested numerically \cite{Grado-White:2024gtx}, and no counterexamples were found so far.
% However, a lot less is known about how holographic entropy inequalities would generalize to time-dependent settings in which the entanglement entropy to the leading order would be  given by the areas of the HRT surfaces. In previous work \cite{}, it was shown that HRT surfaces also obey SSA and MMI and more recent work \cite{Bousso:2024ysg} suggests that the rest of the holographic entropy inequalities should generalize even though they have not yet been proven to hold more generally. 

\paragraph{Generalized Entanglement Wedges} Recently, the notion of entanglement wedge has been extended to arbitrary spacetimes as follows. A \emph{generalized entanglement wedge} can be assigned not only to any boundary subregion in AdS, but to any subregion of any gravitating spacetime~\cite{Bousso:2022hlz,Bousso:2023sya}. This represents a significant expansion of the holographic paradigm, building on previous formulations of the holographic principle in the form of an entropy bound that applies in general spacetimes~\cite{Bousso:1999xy,Bousso:1999cb}. The quantum information in gravitating spacetimes adheres to the holographic principle, regardless of the presence of an asymptotic boundary.

The extension of the concept of entanglement wedges to bulk subregions raises an obvious question: does the HEC, previously derived for boundary subregions, apply to these more general constructs? The boundary of the entanglement can overlap with the boundary of the bulk input region, where it need not be extremal, so existing proofs do not immediately apply. In \cite{Bousso:2024ysg}, we made some progress on this question by showing that MMI also holds for generalized entanglement wedges both in static and time dependent settings. 

An independence condition on the bulk regions enters as an important novel criterion for the validity of MMI for generalized entanglement wedges. This criterion is trivially satisfied for disjoint boundary regions; it captures the nonlocal interdependence of bulk degrees of freedom in quantum gravity. 

In this paper, we will prove that all HEC inequalities satisfied by static entanglement wedges will also hold for static generalized entanglement wedges -- subject to  appropriate independence conditions on the input regions which we will identify. 

\paragraph{Notation} $\EW(B)$ denotes the entanglement wedge of an AdS boundary region; $E(b)$ is the generalized entanglement wedge of a static bulk region $b$. Set inclusion ($\supset$, $\subset$) is consistent with set equality ($=$). Sans serif letters $\mathsf{x},\mathsf{y},\mathsf{f}$ are used to denote classical bit strings; the $r$-th bit in the string $\mathsf{f}$ is $\mathsf{f}_r$. Further notation is defined immediately below.

\section{Static Generalized  Holographic Entropy Cone}

\begin{defn}
    Let $\Sigma$ be a time-reflection symmetric Cauchy slice. An open subset $a$ of $\Sigma$ will be called a \emph{wedge}. (This terminology is chosen for forward compatibility with the time-dependent case not considered in this paper. Otherwise it would be more natural to call $a$ a spatial region.)
\end{defn}

\begin{defn}
For any wedge $a$, $\A(a)$ denotes the area of  $\partial a$, its boundary in $\Sigma$.
\end{defn}
\begin{defn} \label{def:complement} The \emph{complement} of $a$ is the wedge
\begin{equation}
    a'  \equiv \setint (\Sigma\setminus a)~,
\end{equation}
where $\setint$ denotes the interior of a set.
\end{defn}

\begin{defn}
    Let $a,b\subset \Sigma$ be wedges. The \emph{wedge union} of $a$ and $b$ is the double complement of the usual set union:
    \begin{equation}
        a\Cup b\equiv (a\cup b)''~.
    \end{equation}
    The purpose of the double complement is to add in the portions of $\partial a\cap \partial b$ that separate $a$ from $b$. Whenever possible, we use the abbreviated notation
    \begin{equation}
        ab\equiv a\Cup b~.
    \end{equation}
\end{defn}

\begin{defn}
Let $\tilde \Sigma$ denote the \emph{conformal completion} of $\Sigma$~\cite{Wald:1984rg}. The boundary of $\tilde \Sigma$, denoted as $\partial \tilde \Sigma$,  is called the \emph{conformal boundary} of $\Sigma$. 
\end{defn}
\begin{defn}\label{def:staticconformaledge}
The boundary of $a$ in the conformal completion $\tilde \Sigma$ will be denoted $\delta a$.
% As shown in Fig. \ref{fig-conf-bdy}, 
We define the {\em conformal boundary} of $a$ as the set
\begin{equation}
    \tilde\partial a\equiv \delta a\cap \partial\tilde \Sigma~;
\end{equation}
thus,
\begin{equation}
    \delta a = \partial a\sqcup \tilde\partial a~,
\end{equation}
where $\sqcup$ denotes the disjoint union;
see Fig \ref{fig-conf-bdy}.
\begin{figure}[t]
  \centering
 \includegraphics[width = 0.5\linewidth]{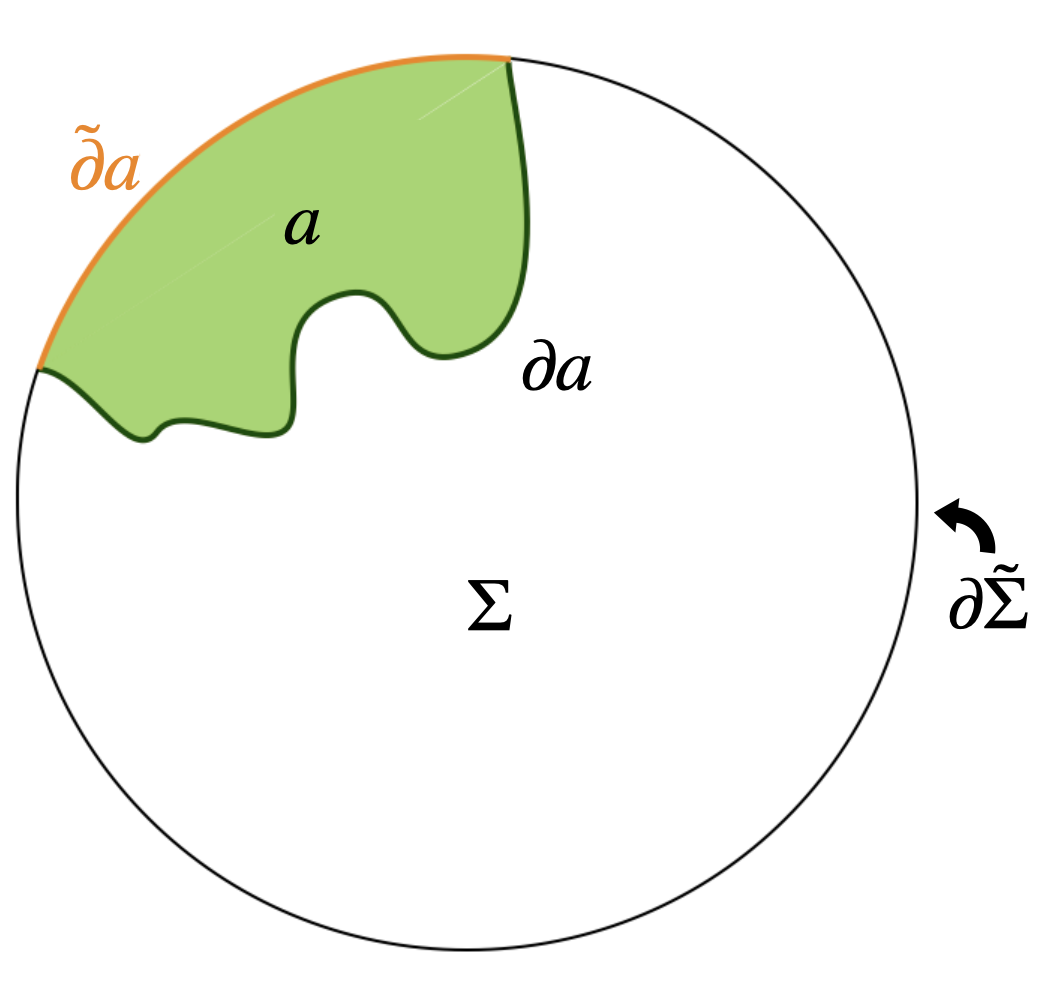}
\caption{Static Cauchy slice $\Sigma$ with conformal boundary $\partial \tilde \Sigma$.
The boundary $\delta a \subset \tilde\Sigma$ of a wedge $a$ in the conformal completion $\tilde\Sigma$ decomposes into the original boundary $\partial a \subset \Sigma$ and a conformal boundary $\tilde \partial a \subset \partial \tilde\Sigma$.}
\label{fig-conf-bdy}
\end{figure}

\end{defn}
\begin{defn}[Static Generalized Entanglement Wedge, Classical Limit~\cite{Bousso:2022hlz}]\label{def:staticgew}
    Let $a$ be a wedge. The static entanglement wedge $E(a)$ is the open subset of $\Sigma$ that contains $a$, has the same conformal boundary as $a$ (if any), and has the smallest boundary area among all such sets.
\end{defn}

%The main result of this paper will be to show that the holographic entropy cone inequalities obeyed by Ryu Takaynagi (RT) surfaces \cite{Bao:2015bfa} also hold for generalized entanglement wedges.

\begin{rem} \label{rem:mmi}
Below we state and prove our main result, building on~\cite{Bao:2015bfa,Akers:2021lms}. As a pedagogical example to guide intuition and illustrate definitions, we will frequently make reference to MMI, 
\begin{multline}\label{eq:mmi}
    \A[E(a_1a_2)] +  \A[E(a_2a_3)] +  \A[E(a_1a_3)]   \\ \geq  \A[E(a_1)] +  \A[E(a_2)] +  \A[E(a_3)] +  \A[E(a_1a_2a_3)] ~,
\end{multline}
even though this is among the two HEC inequalities that have already been proven for generalized entanglement wedges~ \cite{Bousso:2024ysg}. To convert the above expression to the notation in Eq.~\eqref{eq:ineq} below, one should set $L=3$, $v_1 =a_2a_3$, $v_2 =a_1a_3$, $v_3 =a_1a_2$; and $R=4$, $w_1 =a_1$, $w_2 =a_2$, $w_3 =a_3$, $w_4 = a_1a_2a_3$.
    
\end{rem}

\begin{defn}
    Let $\mathsf{x}=(\mathsf{x}_1,\ldots,\mathsf{x}_L) \in\{0,1\}^L$ be a classical bit string of length $L$ and let $\mathsf{y}=(\mathsf{y}_1,\ldots,\mathsf{y}_L)$ be another such string. We will write $\mathsf{x}<\mathsf{y}$ if the inequality holds when the strings are viewed as an integer written in binary, i.e., if 
    \begin{equation}
    \sum_{l=1}^L \mathsf{x}_l 2^l<\sum_{l=1}^L \mathsf{y}_l 2^l~.
    \end{equation}
\end{defn}

\begin{defn}\label{def:distance}
    Let $\mathsf{x}$ and $\mathsf{y}$ be bit strings of length $L$; and let $\alpha=(\alpha_1,\ldots,\alpha_L)\in \mathbb{N}^L$ be a string of positive integers. We define the \emph{$\alpha$-distance} between $\mathsf{x}$ and $\mathsf{y}$ as
     \begin{equation}
    d(\alpha;\mathsf{x},\mathsf{y}) \equiv \sum_{l=1}^L \alpha_l |\mathsf{x}_l-\mathsf{y}_l|~.
\end{equation}
\end{defn}

\begin{defn}\label{def:contraction}
Given strings of positive integers $\alpha=(\alpha_1,\ldots,\alpha_L)\in \mathbb{N}^L$ and $\beta=(\beta_1,\ldots,\beta_R)\in \mathbb{N}^R$, the function $\mathsf{f}:\{0,1\}^L\to\{0,1\}^R$ is an \emph{$\alpha,\beta$ contraction map} if 
\begin{equation}
     d(\alpha;\mathsf{x},\mathsf{y}) \geq d(\beta;\mathsf{f}(\mathsf{x}),\mathsf{f}(\mathsf{y})) 
\end{equation}
for all $\mathsf{x},\mathsf{y} \in \{ 0,1\}^L$. 
\end{defn}

\begin{thm} (Generalized Static Holographic Entropy Cone)\label{thm:staticHEC}
    Let $a_0=\varnothing$; and let $a_1,\ldots, a_n$ be open subsets of the static Cauchy slice $\Sigma$ such that
\begin{equation}\label{eq:staticindep}
    a_i \subset [E(\Cup_{j\neq i} a_j)]'
\end{equation}
for all $i$, where $E$ denotes the generalized entanglement wedge (Def.~\ref{def:staticgew}) and the prime denotes the wedge complement (Def.~\ref{def:complement}). For each $i=0,\ldots,n$, let $\mathsf{x}^i$ be a string of $L$ bits, let $\mathsf{y}^i$ be a string of $R$ bits, with $\mathsf{x}^0=00\ldots 0$ and $\mathsf{y}^0=00\ldots 0$; and let 
\begin{equation}
    v_l=\Cup_{i:x^i_l=1} a_i~,~~~w_r=\Cup_{i:y^i_r=1} a_i~.
\end{equation}
%Let $V_1,\ldots,V_L$ and $W_1,\ldots,W_R$ be $L+R$ mutually distinct subsets of $\{a_1,\ldots,a_n\}$; and let the wedges $v_1, \ldots, v_L$ and $w_1,\ldots, w_L$ be the wedge unions of their elements, i.e., $v_1 = \Cup_{b\in v_1} b$, etc. 
Let $\alpha=(\alpha_1,\ldots,\alpha_L)\in \mathbb{N}^L$ and $\beta=(\beta_1,\ldots,\beta_R)\in \mathbb{N}^R$.
If there exists an $\alpha,\beta$ contraction map $\mathsf{f}:\{0,1\}^L\to\{0,1\}^R$ with $\mathsf{f}(\mathsf{x}^i)=\mathsf{y}^i$ for all $i\in \{0,\ldots,n\}$, then
\begin{equation}
\label{eq:ineq}
    \sum_{l=1}^L \alpha_l \A[E(v_l)] \ge \sum_{r=1}^R \beta_r \A[E(w_r)] ~.
\end{equation}
\end{thm}

\begin{proof}
%The proof of this theorem follows from a slight generalization of the proof by contraction idea \cite{Bao:2015bfa} to applied to entanglement wedges of gravitating subregions. In the following we will follow in the footsteps of the discussion in \cite{Akers:2021lms} and carry out the generalization to the case at hand.

\begin{figure}
\includegraphics[width=10cm]{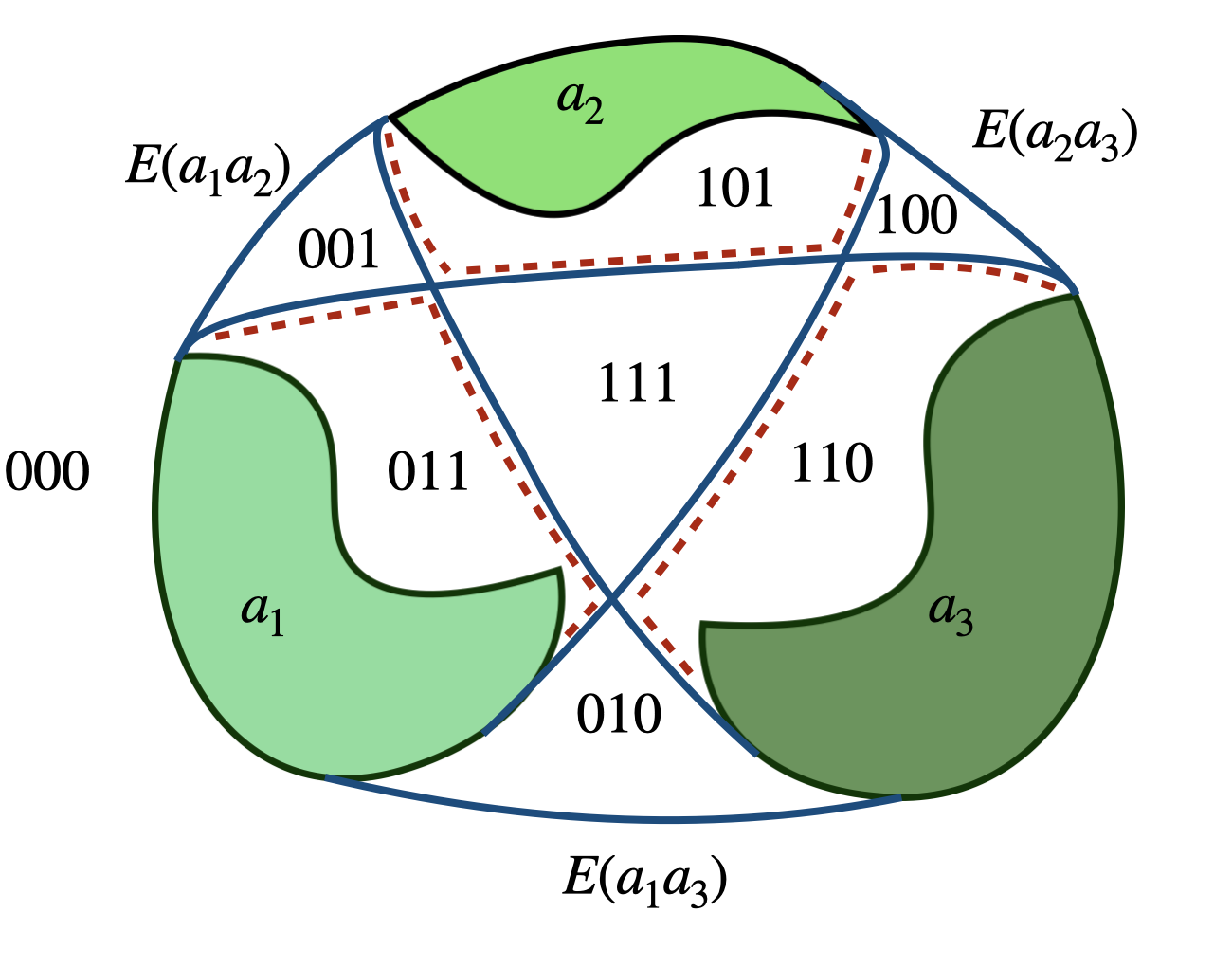}
\centering
\caption{\label{ex:mmi} Illustration of various definitions using the example of MMI (see Remark~\ref{rem:mmi}). The Cauchy slice $\Sigma$ is partitioned into tiles $s(\mathsf{x})$ labeled by strings of 3 bits $\mathsf{x}$ (shown) that describe their inclusion (1) or exclusion (0) in the three entanglement wedges appearing on the left hand side of Eq.~\eqref{eq:mmi}. The unique tiles containing a single set $a_i$ carry the labels $\mathsf{x}^1 = 011$, $\mathsf{x}^2 = 101$, and $\mathsf{x}^3 = 110$.
Candidate sets $c_r$ for $E(w_r)$ are built from these tiles; thus if some $a_i\subset w_r$, then the corresponding tile must be included. For example, since $w_2 = a_2$, $c_2$ must include $s(101)$; this implies $\mathsf{f}_2(101)= 1$ (see Table \ref{tab:mmi}). Similarly, to ensure the correct conformal boundary some tiles must be excluded. For example, the tile $s(000)$ must be excluded from all $c_r$; this implies $\mathsf{f}(\mathsf{x}^0)= \mathsf{f}(000)= 0000$.}
\end{figure}

We begin by partitioning\footnote{The tiles partition $\Sigma$ in the sense that $s(\mathsf{x})\cap s(\mathsf{y}) = \delta_{\mathsf{xy}}$; and that $\Sigma=\Cup_{\mathsf{x}}s(\mathsf{x})$.} $\Sigma$ into $2^L$ ``tiles'' $s(\mathsf{x})$, characterized by their inclusion in, or exclusion from, each entanglement wedge $E(v_l)$:
    \begin{equation}
\label{eq:sigmaxs}
    s(\mathsf{x}) = \bigcap_{l=1}^L E(v_l)^{\mathsf{x}_l}~, ~~\text{where} ~~ E(v_l)^{j} =
    \begin{cases}
        E(v_l)  &\text{if} \quad j=1~, \\
        E(v_l)' &\text{if} \quad j=0~. 
    \end{cases}
\end{equation}
\begin{table}%[h!]
\setlength{\tabcolsep}{.3cm}
	\centering
	\begin{tabular}{c || c | c | c||c | c | c | c}
		& $a_2a_3$ & $a_1a_3$ & $a_1a_2$ & ~$a_1$~ & ~$a_2$~ & ~$a_3$~ & $a_1a_2a_3$ \\ \hline
        & $\mathsf{x}_1$ & $\mathsf{x}_2$ & $\mathsf{x}_3$ & ~$\mathsf{f}_1(\mathsf{x})$~ & ~$\mathsf{f}_2(\mathsf{x})$~ & ~$\mathsf{f}_3(\mathsf{x})$~ & $\mathsf{f}_4(\mathsf{x})$ \\ \hline
		$\mathsf{x}^0$ & 0 & 0 & 0 & 0 & 0 & 0 & 0 \\
		 ~  & 0 & 0 & 1 & 0 & 0 & 0 & 1 \\
		 ~  & 0 & 1 & 0 & 0 & 0 & 0 & 1 \\
         $\mathsf{x}^1$ & 0 & 1 & 1 & 1 & 0 & 0 & 1 \\
		 ~  & 1 & 0 & 0 & 0 & 0 & 0 & 1 \\
		$\mathsf{x}^2$ & 1 & 0 & 1& 0 & 1 & 0 & 1 \\

		$\mathsf{x}^3$ & 1 & 1 & 0 & 0 & 0 & 1 & 1 \\
		 ~  & 1 & 1 & 1 & 0 & 0 & 0 & 1 \\
	\end{tabular}
	\caption{
	The contraction map $\mathsf{f}:\{0,1\}^3\to\{0,1\}^4, \mathsf{x}\to \mathsf{f}(x)$ that implies MMI for generalized entanglement wedges. The top row shows the input wedges $v_l$ and $w_r$. The truth values of $\mathsf{x}$ correspond to the inclusion (1) or exclusion (0) of each left-hand-side entanglement wedge from the tile $s(\mathsf{x})$. The independence condition \eqref{eq:staticindep} ensures that each input region $a_i$ is contained in exactly one tile $s(\mathsf{x}^i)$; see the leftmost column. The truth value of $\mathsf{f}_r(\mathsf{x})$ determines whether the tile $s(\mathsf{x})$ is used as part of the region $c_r$.} 
	\label{tab:mmi}
\end{table}
and  $\mathsf{x}=(\mathsf{x}_1,\ldots,\mathsf{x}_l) \in\{0,1\}^L$; see Fig.~\ref{ex:mmi}. The tiles will have shared boundary segments
\begin{equation}
\label{eq:adjacent}
    \gamma(\mathsf{x},\mathsf{y}) = \partial s(\mathsf{x}) \cap \partial s(\mathsf{y})~,~~~\mathsf{x}<\mathsf{y}~,
\end{equation}
where we restrict to $\mathsf{x}<\mathsf{y}$ to avoid overcounting.
Note that
\begin{equation}
\label{eq:gammapiec}
    \partial E(v_l) = \bigcup_{\mathsf{x}<\mathsf{y}:\mathsf{x}_l \ne \mathsf{y}_l} \gamma(\mathsf{x},\mathsf{y})~;
\end{equation}
and hence,
\begin{equation}\label{static_LHS}
    \A[E(v_l)] = \sum_{\mathsf{x}<\mathsf{y}}\, |\mathsf{x}_l - \mathsf{y}_l|\, \A\left[\gamma(\mathsf{x},\mathsf{y})\right]~.
\end{equation}
Using Def.~\ref{def:distance} this implies
\begin{equation} \label{LHS}
     \sum_{l=1}^L \alpha_l \A[E(v_l)] 
     = \sum_{\mathsf{x}<\mathsf{y}} 
d(\alpha;\mathsf{x},\mathsf{y})\A\left[\gamma(\mathsf{x},\mathsf{y})\right]~.
\end{equation}

We now define $R$ new wedges as unions of tiles:
\begin{equation}
\label{eq:rhshomo}
    c_r = \bigcup_{\mathsf{x}:\,\mathsf{f}_r(\mathsf{x})=1} s(\mathsf{x})~.
\end{equation}
Thus, the $r$-th truth value of the contraction map $\mathsf{f}$ 
informs us whether the tile $s(\mathsf{x})$ is included in $c_r$. We note that
\begin{equation} \label{static_RHS}
    \A[c_r] =  \sum_{\mathsf{x}<\mathsf{y}} \, |\mathsf{f}_r(\mathsf{x}) - \mathsf{f}_r(\mathsf{y})|\, 
    \A\left[\gamma(\mathsf{x},\mathsf{y})\right].
\end{equation}
Using Def.~\ref{def:distance} this implies
\begin{equation} \label{LHS}
     \sum_{r=1}^R \beta_r \A[c_r] 
     = \sum_{\mathsf{x}<\mathsf{y}} 
d(\beta;\mathsf{f}( \mathsf{x}), \mathsf{f}(\mathsf{y}))\A\left[\gamma(\mathsf{x},\mathsf{y})\right]~.
\end{equation}
Since $\mathsf{f}$ is an $\alpha,\beta$ contraction map (see Def.~\ref{def:contraction}), it follows that
\begin{equation}
    \sum_{l=1}^L \alpha_l \A[E(v_l)] \geq \sum_{r=1}^R \beta_r \A[c_r]~. 
\end{equation}

By the independence condition, Eq.~\eqref{eq:staticindep}, $s(\mathsf{x}^i)$ is the unique tile that contains $a_i$ and no other tiles intersect with $a_i$. The assumption that 
$\mathsf{f}(\mathsf{x}^i)=\mathsf{y}^i$ thus ensures that $c_r$ contains precisely the $a_i$'s that constitute $w_r$. This implies that $c_r\supset w_r$ and that $\tilde\partial c_r=\tilde\partial w_r$. The latter equality follows because only the $n+1$ tiles $s(\mathsf{x}^i)$ can have a conformal boundary, by Eq.~\eqref{eq:staticindep} and Def.~\ref{def:staticgew}; the assumption that $\mathsf{f}(\mathsf{x}^0)=\mathsf{y}^0$ is vital.

By Def.~\ref{def:staticgew}, $E(w_r)$ is the smallest-area wedge that contains $w_r$ and has the same conformal boundary. We have just established that $c_r$ is among the wedges satisfying these criteria, so it follows that $\A[E(w_r)]\leq \A(c_r)$. Combined with the above inequality this yields Eq.~\eqref{eq:ineq}.
\end{proof}

\section{Discussion}

The proof of the traditional (boundary input region) static holographic entropy cone \cite{Bao:2015bfa, Akers:2021lms} differs from the present proof of the generalized static entropy cone differ chiefly through the role of the independence condition~\eqref{eq:staticindep}. Both proofs rely on the existence of a contraction map $\mathsf{f}$ to construct regions that satisfy the definition of an entanglement wedge except for area minimalization. 

\begin{figure}[ht!]
\includegraphics[width=10cm]{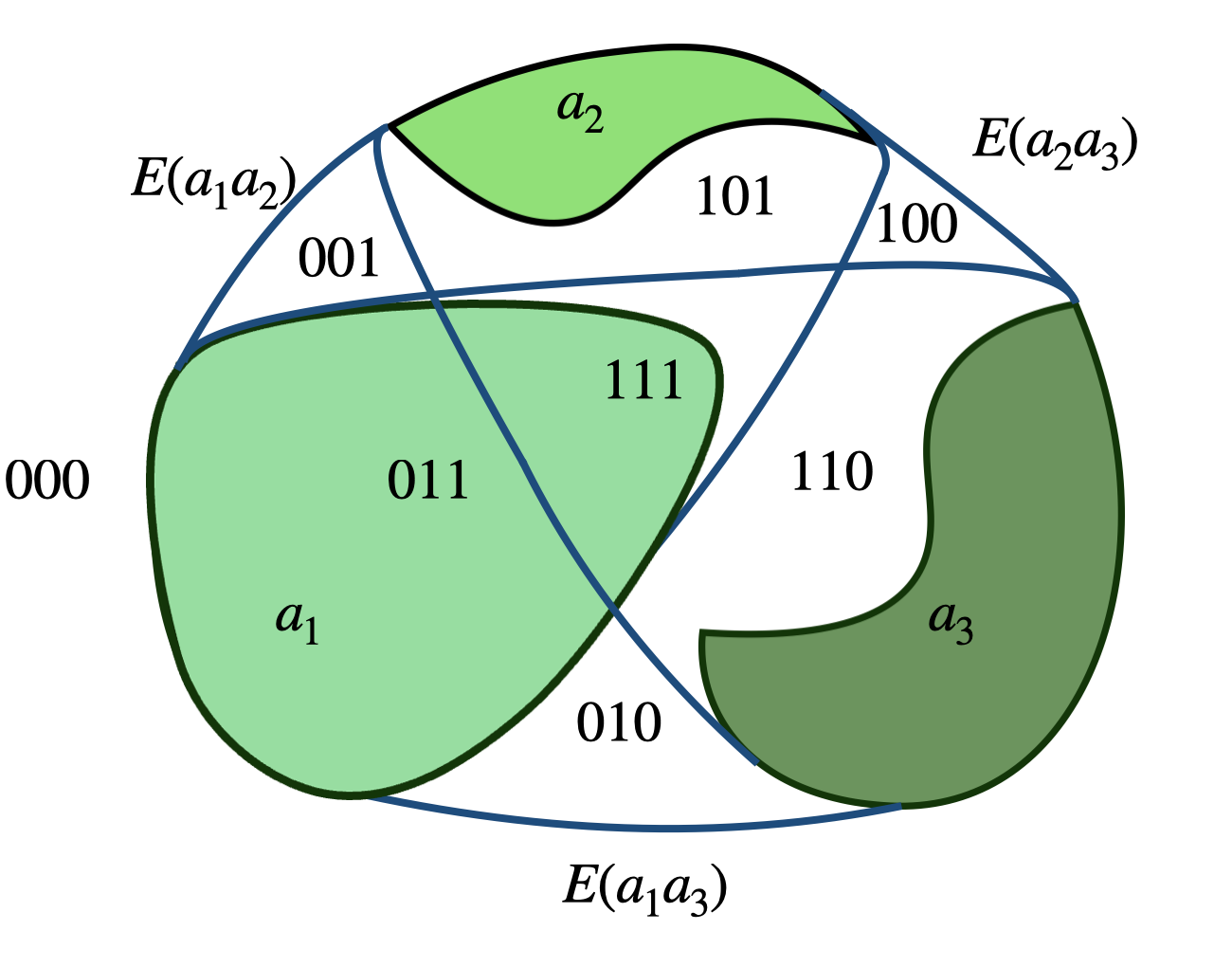}
\centering
\caption{\label{ex:mmi-counter} An example of a configuration that can violate MMI because the input regions violate independence condition \eqref{eq:staticindep}: $a_1 \not \subset E(a_2a_3)'$.}
\end{figure}

Technically, the independence condition ensures that each bulk input region $a^i$ is fully contained in one tile $s(\mathsf{x}^i)$ and does not intersect any other tile. (With boundary input regions, this followed automatically from the homology constraint on RT surfaces.) If this was not the case, then the inclusion condition of Def. \ref{def:staticconformaledge} would impose additional constraints on the map $\mathsf{f}$ that are not satisfied by the known families of contraction maps. In fact, it is possible to find configurations which do not obey the independence condition \ref{eq:staticindep} and violate the HEC; see Fig.~\ref{ex:mmi-counter} for an example. 

We expect that the true origin of the independence condition is deeper. The existence of bulk entanglement wedges illustrates that there is a sense in which fundamental degrees of freedom in the input region contain information about the semiclassical state of the entanglement wedge~\cite{Bousso:2022hlz,Bousso:2023sya}. 

This can be made more precise is if the input region $a$ is an asymptotic region in AdS, with conformal boundary $\tilde \partial a$. In that case it is known that the local CFT operators in $\tilde \partial a$ correspond to quasilocal semiclassical operators $\{\cal O\}$ in $a$. To reconstruct $E(a)$, however, one needs the full operator algebra; this is generated by the operators in $\{\cal O\}$. But generic operators in the full algebra have no semiclassical interpretation in $a$.  

Thus, the condition \ref{eq:staticindep} corresponds to the statement that no data in the input region $a_i$ is accessible to the fundamental degrees of freedom associated with the union of all other $a_j$, $j\neq i$. We refer the reader to upcoming work \cite{kaya2025} for a further discussion of the independence condition  with examples in tensor networks and fixed induced metric states.

\subsection*{Acknowledgements}
We thank Pratik Rath for discussions. This work was supported in part by the Berkeley Center for Theoretical Physics; and by the Department of Energy, Office of Science, Office of High Energy Physics under Award DE-SC0025293.
\bibliographystyle{JHEP}
\bibliography{generalizedHEC}

\end{document}